\documentclass{article}%{revtex4}

\usepackage{amsmath,amsthm}
\usepackage[superscript]{cite}
\usepackage{authblk}
\newtheorem{proposition}{Proposition}[section]
\newtheorem{corollary}{Corollary}[section]
\newtheorem{lemma}{Lemma}[section]
\newtheorem{theorem}{Theorem}[section]
\newtheorem*{note}{Note}

\providecommand{\keywords}[1]{\textbf{\textit{Keywords--}} #1}

\begin{document}

\title{On the long-time asymptotics of quantum dynamical semigroups}

\author{G. A. RAGGIO$^*$ and P. R. ZANGARA}

\affil{FaMAF, Universidad Nacional de C\'ordoba,\\ C\'ordoba, C\'ordoba X5000, Argentina\\
$^*$E-mail: raggio@famaf.unc.edu.ar}

%\author{G. A. RAGGIO$^*$ and P. R. ZANGARA\\
%FaMAF, Universidad Nacional de C\'ordoba,\\ C\'ordoba, C\'ordoba X5000, Argentina\\
%$^*$E-mail: raggio@famaf.unc.edu.ar}

\date{}

\maketitle

\begin{abstract}
We consider semigroups $\{\alpha_t: \; t\geq 0\}$ of normal, unital,  completely positive maps $\alpha_t$ on a von Neumann algebra ${\mathcal M}$. The (predual) semigroup $\nu_t (\rho ):= \rho \circ \alpha_t$ on normal states $\rho$  of $\mathcal M$ leaves invariant the face ${\mathcal F}_p:= \{\rho : \; \rho (p)=1\}$ supported by the projection $p\in {\mathcal M}$, if and only if $\alpha_t(p)\geq p$ (i.e., $p$ is sub-harmonic). We complete the arguments showing that the sub-harmonic projections form a complete lattice. We then consider $r_o$, the smallest projection which is larger than each support of a minimal invariant face; then $r_o$ is subharmonic. In finite dimensional cases   $\sup \alpha_t(r_o)={\bf 1}$ and $r_o$ is also the smallest projection $p$ for which $\alpha_t(p)\to {\bf 1}$. If  $\{\nu_t: \; t\geq 0\}$  admits a faithful family of normal stationary states then $r_o={\bf 1}$ is useless; if not, it helps to reduce the problem of the asymptotic behaviour of the semigroup for large times.
\end{abstract}


%\bodymatter

\section{Introduction and preliminaries}

We consider a von Neumann algebra $\mathcal M $ and denote its normal state space by $\mathcal S$. A \textit{quantum dynamical semigroup} $\{\alpha_t:\; t\geq 0\}$ is a family of normal, unital, positive, linear maps $\alpha_t:{\mathcal M} \to {\mathcal M}$ with the property $\alpha_t\circ \alpha_s = \alpha_{t+s}$ where $\alpha_0$ is the identity.  Then, the map $\nu_t: {\mathcal S}\to {\mathcal S}$ defined by $\nu_t( \rho )= \rho \circ \alpha_t$ is affine, $\nu_0$ is the identity, and $\nu_t\circ \nu_s=\nu_{t+s}$. Conversely, given a semigroup $\{\nu_t : \; t\geq 0\}$ of affine maps on $\mathcal S$, the dual maps are a positive quantum dynamical semigroup.\\

One often demands on physical grounds, that  $\alpha_t$ be completely positive. 
When ${\mathcal M}$ is the algebra of all bounded linear operators on a Hilbert space, if the dynamical semigroup is strongly continuous in $t$, and each $\alpha_t$ is completely positive, the generator has the canonical GKS-Lindblad form.\\

The long-time asymptotics of such semigroups has been studied  in the 1970's and in the 1980's, after pioneering papers of E.B. Davies \cite{Da}\,, culminating with the work of Frigerio\cite{Fri1,Fri2,FrVe}\,, and U. Groh\cite{Gr}\,. More recent studies are due to  Fagnola \& Rebolledo\cite{FaRe1,FaRe2,FaRe3}\,, Umanit\'a\cite{Um}\, ,  Mohari\cite{Mo1,Mo2} \,  and Baumgartner \& Narnhofer\cite{BaNa}\,. We refer to Ref.~\citen{FaRe3} for a recent overview.
In pertinent cases,  the asympotics can be studied via  the GKS-Lindblad generator.\\

In this note, all projections are ortho-projections (self-adjoint equal to its square). $\mathbf 1$ denotes the identity operator and for a projection $p$, $p^{\perp}={\mathbf 1}-p$. Limits in $\mathcal M$ are invariably in the $w^*$-topology. All states (positive linear functionals of unit norm) are normal. Limits of states are with respect to the distance induced by the norm. But recall that the norm-closure of a convex set of states coincides with its weak-closure. The \textit{support} of a state $\rho$ --written $s_{\rho}$-- is the smallest projection $p\in {\mathcal M}$ such that $\rho (p)=1$.\\
We will consider a quantum dynamical semigroup $\{\alpha_t:\; t\geq 0\}$ and will always explicitly mention any additional positivity hypotheses.
In particular,   if each $\alpha_t$ is completely positive, we say that the semigroup is CP. A state $\omega$ is \textit{stationary} if $\nu_t(\omega)=\omega \circ \alpha_t=\omega$.

\section{Invariant faces and sub-harmonic projections}

A \textit{face} is a convex subset ${\mathcal F}$ of ${\mathcal S}$ which is stable under convex decomposition: if $ t \rho + (1-t) \mu \in {\mathcal F}$ for $0<t<1$ with $\rho ,\mu \in {\mathcal S}$ then $\rho, \mu \in {\mathcal F}$. If $p\in {\mathcal M}$ is a projection then ${\mathcal F}_p:= \{ \rho \in {\mathcal S}:\; \rho (p)=1\}$ is a closed face; we say it is the face \textit{supported} by $p$.
It is not so obvious but true \cite{Sa,AsEl}\,, that every closed face is of this form; i.e. it is the face supported by some projection.
Clearly, ${\mathcal F}_p\subset {\mathcal F}_q$ if and only if $p\leq q$.\\

The following result is implicit or partially explicit in the work of Fagnola \& Rebolledo and Umanit\'a.

\begin{proposition} Suppose $\nu$ is an affine map of  $\mathcal S$ into
itself and let $\alpha$ be the dual normal, linear, positive map
of $\mathcal M$ into itself. For a projection $p\in {\mathcal M}$ the following conditions are equivalent:
(1) the face ${\mathcal F}_p$ supported by $p$ is $\nu$-invariant; (2) $\alpha (p)\geq p$;
(3) $p\alpha (a)p=p\alpha (pap)p$ fore every $a\in {\mathcal M}$; (4) $\alpha (p^{\perp}ap^{\perp})= p^{\perp}\alpha (p^{\perp}ap^{\perp})p^{\perp}$ for every $a\in {\mathcal M}$.
\end{proposition}

\begin{proof} We first prove the chain $(2)\Rightarrow (1)\Rightarrow (3)\Rightarrow (2)$. If $\alpha (p) \geq p$ then, for any state $\rho$ one has
$\nu (\rho) (p) = \rho (\alpha (p)) \geq \rho (p)$. Thus, $\rho
(p)=1$ implies $\nu ( \rho ) (p) )=1$, i.e. $\nu ({\mathcal
F}_p)\subset {\mathcal F}_p$.
If $\nu ({\mathcal F}_p)\subset {\mathcal F}_p$, we show that
\[ (*)\;\;\;
\omega (p\alpha (pap)p)=\omega (p\alpha (a)p) \; \mbox{ for every $\omega \in
{\mathcal S}$}\;.\label{a}\]
Since every normal linear functional is the linear combination of at most four states, this then implies that $p\alpha (pap)p=p\alpha (a)p$.
To prove (*) observe that, by the Cauchy-Schwarz inequality
for states, the claim is trivially valid if $\omega (p)=0$. Otherwise,
consider the state $\omega_p(a):= \omega (pap)/\omega (p)$.
Clearly $\omega_p \in {\mathcal F}_p$; thus,
\[ \omega (p)^{-1}\omega (p\alpha (pap)p)= \omega_p ( \alpha (pap))=
\nu ( \omega_p ) (pap)= \nu (\omega_p)(a)\]
\[ =\omega_p (\alpha (a))=\omega (p)^{-1}\omega (p\alpha (a) p)\;;\]
 which is (*).
Finally, if $p\alpha (pap)p=p\alpha (a) p$, then $p - p\alpha (p)p=p\alpha (p^{\perp})p=0$ and Lemma \ref{pxp=p} of the Appendix implies $\alpha (p)\geq p$.\\
If $0\leq x=p^{\perp}xp^{\perp}\leq {\bf 1}$ then, by Lemma \ref{pxp=p} of the Appendix, $x\leq p^{\perp}$ and $\alpha (x) \leq \alpha (p^{\perp})$; when (2) is the case $\alpha (p^{\perp}) \leq p^{\perp}$ so that $\alpha (x)\leq p^{\perp}$ which by the aforementioned Lemma, implies $p^{\perp}\alpha (x)p^{\perp}=\alpha (x)$. For general $0\leq x=p^{\perp}xp^{\perp}$ we consider $x/\|x\|$ and obtain $p^{\perp}\alpha (x)p^{\perp}=\alpha (x)$. Since every $a\in {\mathcal M}$ is a linear combination of at most four positive elements, we conclude that (2) implies (4). But (4) implies $\alpha (p^{\perp})=p^{\perp}\alpha (p^{\perp})p^{\perp}$ which, by the same Lemma, gives $\alpha (p^{\perp} ) \leq p^{\perp}$ which is equivalent to $\alpha (p) \geq p$.
\end{proof}

In the context of quantuym dynamical semigroups, a projection $p$ satisfying $\alpha_t(p)\geq p$ has
been termed \textit{sub-harmonic} \cite{FaRe1}\,. We say the
projection $p$ is \textit{sub-harmonic} for the linear, normal,
unital and positive map $\alpha$ on $\mathcal M$ if $\alpha (p)
\geq p$. The previous proposition relates the sub-harmonic property of a
projection to the more geometric notion of invariance of the
supported face. This relationship can be immediately put to use:

\begin{proposition} If a family of projections is sub-harmonic for a linear, normal, unital and positive map $\alpha$ on $\mathcal M$,  then the infimum  of the family is sub-harmonic for $\alpha$.
 \end{proposition}

\begin{proof} If $\{{\mathcal F}_{\iota}:\; \iota \in I\}$ is a family of closed faces ${\mathcal F}_{\iota}$ of $\mathcal S$ then $\bigcap_{\iota}
{\mathcal F}_{\iota}$ is  a closed face and it is the largest closed face contained in each ${\mathcal F}_{\iota}$.
The support of $\bigcap_{\iota} {\mathcal F}_{\iota}$ is exactly $\inf \{ p_{\iota}: \; \iota \in I\}$,
where $p_{\iota}$ is the support of ${\mathcal F}_{\iota}$.   Moreover, if each ${\mathcal F}_{\iota}$ is $\nu$-invariant then so is the intersection.
\end{proof}

The corresponding statement for the supremum of such a family  has
been observed and proved directly (Ref.~\citen{Um}).\\
 For
projections $p$ that are \textit{super-harmonic}, i.e. $\alpha (p)
\leq p$ or equivalently $p^{\perp}$ is sub-harmonic, we have (in
reply to a question posed in Ref.~\citen{Fa}):
\begin{corollary}If a family of projections is super-harmonic for a linear, normal, unital and positive map $\alpha$ on $\mathcal M$,  then the
supremum of the family is super-harmonic for $\alpha$.
 \end{corollary}

 \begin{proof} $\sup \{ p :
\; p\in {\mathcal F}\}= (\inf\{ p^{\perp}: \; p\in {\mathcal
F}\})^{\perp}$ and $\inf\{ p^{\perp}: \; p\in {\mathcal F}\}$ is
sub-harmonic by the previous proposition.
\end{proof}
The corresponding statement for the infimum of a super-harmonic
family follows from the result for the supremum of a sub-harmonic
family by orthocomplementation as above. Thus,
\begin{theorem} The set of sub-harmonic and the set of super-harmonic
projections with respect to a linear, normal, unital and positive
map on $\mathcal M$ are both complete lattices.
\end{theorem}

A \textit{minimal invariant face} is a closed $\nu_t$-invariant
face which does not properly contain another non-empty closed
$\nu_t$-invariant face. Equivalently, it is a face whose support
is a minimal sub-harmonic projection, i.e. a sub-harmonic projection
that is not larger than a non-zero sub-harmonic projection other
than itself. One can prove, and this goes back to --at least--
Davies (see Ref.~\citen{Da}, Theorem 3.8 of Sect. 6.3), that if the
minimal invariant face admits a stationary state then it is unique and its
support is the support of the face. Moreover (Ref.~\citen{Gr},
Proposition 3.4) the restriction of $\nu_t$ to the face is ergodic
(the Ces\`aro means converge to the stationary state).

\subsection*{A ``recurrent'' projection}
We define the {\em minimal recurrent} projection $r_o$ as the smallest projection which is larger than every minimal sub-harmonic projection.
 Equivalently, $r_o$ is the support of the smallest $\nu_t$-invariant face which contains every minimal $\nu_t$-invariant face.
By virtue of its definition and the  result mentioned above --to
the effect that the supremum of a family of sub-harmonic
projections is sub-harmonic-- it follows that the minimal recurrent
projection is sub-harmonic. Hence the directed family
$\alpha_t(r_o)$ which is bounded above by ${\mathbf 1}$ has a
lowest upper bound in ${\mathcal M}$ denoted by $x$ which is
positive and below $\mathbf 1$. Since $x= \lim_{t\to \infty}
\alpha_t(r_o)$ it follows that $\alpha_t(x)=x$ for every $t\geq
0$. Let $s[x]$ denote the support of $x$, that is the smallest
projection $p\in {\mathcal M}$ with $xp=x$. The following
treatment follows the lines of work by Mohari \cite{Mo2}\,.

\begin{lemma} If $\{\alpha_t:\;t\geq 0\}$ is CP, then $s[x]=\mathbf 1$.
\end{lemma}

\begin{proof} $s[x]^{\perp}$ is the largest projection  $q$ with
$xq=0$, and it is sub-harmonic by a result of Ref.~\citen{Mo2} quoted
in the appendix. Assume that  $s[x]\neq \mathbf 1$; then there is a minimal 
sub-harmonic non-zero projection $q$ with $q\leq s[x]^{\perp}$. One
has $xq=0$. By the definition of $r_o$, we have $q\leq r_o$ and
thus $q=qr_oq\leq q\alpha_t(r_o)q\leq qxq =0$,
which contradicts the assumption.
\end{proof}

Let ${\mathcal J}:= \{ a\in {\mathcal M}:\; \lim_{t\to\infty} \alpha_t(a^*a)=0\}$.
Since for each state $\rho$, one has the Cauchy-Schwarz inequality
\[|\rho ( \alpha_t(a^*b^*))|=|\rho (\alpha_t (ba))|=|\nu_t ( \rho )(ba)|\]
\[ \leq \sqrt{
\nu_t(\rho ) (bb^*)\nu_t (\rho ) (a^*a)} \leq \|b\| \sqrt{ \rho (
\alpha_t(a^*a))}\;;\]
we infer that ${\mathcal J}$ is a linear subspace of $\mathcal M$. If
$c\in {\mathcal M}$ and $a\in {\mathcal J}$, the same inequality applied to $b=a^*c^*c$ shows that  $ca\in \mathcal J$; thus ${\mathcal J}$ is a left-ideal. \\

If $\mathcal M$ is finite-dimensional (that is *-isomorphic to the direct sum of finitely many full matrix algebras) then, on the one hand  $s[x]=\mathbf 1$ implies that $x$ is invertible, and Mohari \cite{Mo2} has shown that if $x$ is invertible then $x=\mathbf 1$; and --on the other hand-- $\mathcal J$ is closed and there is  a projection such that $\mathcal J= \mathcal M\cdot z$. Then
\begin{theorem}
 If $\mathcal M$ is finite dimensional and $\{\alpha_t:\;t\geq 0\}$ is CP, then $\sup \{ \alpha_t (r_o): \; t\geq 0\} = {\mathbf 1}$. Moreover ${\mathcal J}={\mathcal M}\cdot r_o^{\perp}$ and $r_o$ is the smallest projection $p\in \mathcal M$ with $\lim_{t\to \infty}\alpha_t(p)=\mathbf 1$.
\end{theorem}

\begin{proof}   There is\cite{Sa}\, a projection $z\in\mathcal M$ with $\mathcal J=\mathcal M\cdot z$. Lemma \ref{pxp=p} implies that $x$ is invertible and Theorem 2.5 of Ref.~\citen{Mo2} gives $x=\mathbf 1$. Hence $r_o^{\perp} \in \mathcal J$ and thus $r_o^{\perp}\leq z$ or $r_o\geq z^{\perp}$. Suppose $p$ is a minimal sub-harmonic projection; there is a  stationary state $\omega$ in the minimal invariant face supported by $p$ and it follows (see the introduction) that it is unique and $s(\omega )=p$. Since $\omega (z)= \omega ( \alpha_t (z))\to 0$, we have $\omega (z^{\perp})=1$ and thus $p\leq z^{\perp}$. But then, by the definition of $r_o$, $r_o \leq z^{\perp}$. Thus $r_o=z^{\perp}$.
\end{proof}

\begin{note} Despite the claim in Ref.~\citen{BaNa}, p. 8, one cannot conclude from $\lim_{t\to \infty}\alpha_t (p)=\mathbf 1$ for a projection $p$, that $p$ is sub-harmonic. Simple examples can be given\cite{Za}\,.
 \end{note}

It follows from the Cauchy-Schwarz inequality for states that $\lim_{t\to \infty} \alpha_t(ar_o^{\perp})=\lim_{t\to \infty} \alpha_t(r_o^{\perp}a)= 0$ for every $a \in {\mathcal M}$ so that $\alpha_t (a) \asymp \alpha_t(r_oar_o)$ for large $t$ and every $a\in {\mathcal M}$. If $\{\nu_t:\; t\geq 0\}$ admits a faithful family of stationary states, then the minimal recurrent projection is the identity. This happens because for a stationary state $\omega$, one has $\omega (r_o^{\perp})= \omega ( \alpha_t(r_o^{\perp}))\downarrow 0$ and thus $\omega ( r_o^{\perp})=0$. However in this case there are results \cite{FrVe,Gr,Mo1} on the asymptotic behaviour of the semigroup. \\
Other recurrent projections have been considered. For example
(Ref.~\citen{Gr}, p. 407; Ref.~\citen{Um}), the supremum $r$ of the supports
of the stationary states (if any are available), which is then
sub-harmonic and above $r_o$.\\

There is no reason to expect that the above theorem holds in infinite dimension.

\section*{Acknowledgements}
We thank the organizers of the 30th Conference on Quantum
Probability and Related Topics, in Santiago, Chile. The support of
CONICET (PIP  11220080101741) is acknowledged. The first author is grateful
to M.E. Mart\'{\i}n Fern\'andez for generous support.

\section{Appendix}

We collect here  two  technical results used in the above proofs.

\begin{lemma} \label{pxp=p}
 For $x\in {\mathcal M}$ satisfying ${\mathbf 1} \geq x \geq 0$ and $p\in {\mathcal M}$  a projection one has:
\begin{itemize}
 \item[a)] the following five conditions are  equivalent: (1) $x\geq p$; (2)$pxp=p$; (3) $x=p+p^{\perp}xp^{\perp}$; (4)$ xp=p$; (5) $px=p$.
\item[b)]  the following four conditions are equivalent: (1) $p\geq x$; (2) $pxp=x$;  (3) $x=xp $; (4) $ x=px$ .
\end{itemize}
\end{lemma}

\begin{proof} a): Given ${\mathbf 1} \geq x \geq p$, multiplication from left and right by $p$ gives $p\geq pxp\geq p$ and thus $pxp=p$.\\
If $pxp=p$ then $p( {\mathbf 1} -x)p=0$ which implies $( {\mathbf
1} -x)^{1/2}p=0$ and thus $( {\mathbf 1} -x)p=0$ or $xp=p$; taking
adjoints $p=px$.\\
 And $xp=p$ or $px=p$ implies $pxp=p$.\\
 Finally either of the equivalent conditions (4) or (5) imply that
 $  x-p= p^{\perp}xp^{\perp}\geq 0$.\\

b):  $p\geq x$ if and only if $p^{\perp}\leq {\mathbf 1}-x$. Apply
a).
\end{proof}

The following crucial observation and the proof, repeated here for convenience, are due to Mohari\cite{Mo2}\,.

\begin{proposition} [Mohari] Suppose $\alpha: {\mathcal M}\to {\mathcal M}$ is linear, unital, normal and completely positive and
 $x \in {\mathcal M}$ is positive with $\alpha (x)=x$. Then the support of $x$ is super-harmonic.
 \end{proposition}
\begin{proof}
 We may assume ${\mathcal M}$ is a von Neumann algebra on a Hilbert space ${\mathcal K}$. By the Stinespring Representation Theorem there is
 a normal $*$-homomorphism $\pi$ of $\mathcal M$ into ${\mathcal B}({\mathcal H})$ (the algebra of bounded linear operators on a Hilbert
 space ${\mathcal H}$) and an isometry $V: {\mathcal K} \to {\mathcal H}$ such that $\alpha (a)= V^*\pi (a) V$ for all $a \in {\mathcal M}$.
 Recall that the support of a self-adjoint element $a$ is the smallest projection $p$ of ${\mathcal M}$ such that $pa=a$ (equivalently $ap=a$).
If $\mathcal M \subset {\mathcal B}({\mathcal K})$ then the support coincides with the smallest projection $q \in {\mathcal B}({\mathcal K})$
such that $qa=a$ (Proposition 1.10.4 of \cite{Sa}). Now if $x$ satisfies the hypothesis, $s$ is its support and $z=s^{\perp}$,
then $0=zxz= z\alpha (x) z= zV^*\pi (x) Vz = (yVz)^*(yVz)$ where $y= \sqrt{\pi (y)}$. Thus $yVz=0$ and hence $\pi (x) Vz=0$.
The support of $\pi (x) $ is $\pi (s)$ and since $Vz$ maps ${\mathcal K}$ into the kernel of $\pi (x)$, we conclude
that $\pi (s) Vz=0$. But then, $\alpha (s) z= V^*\pi (s) Vz =0$ or $\alpha (s)=\alpha (s) s$ which by the Lemma above
implies $\alpha (s) \leq s$.
\end{proof}


\begin{thebibliography}{10}


\bibitem{Da} E.B. Davies: \textit{Quantum Theory of Open Systems.} (Academic Press, London 1976).

\bibitem{Fri1} A. Frigerio: \textit{Quantum dynamical semigroups and approach to equilibrium.}
{\em Lett. Math. Phys.} \textbf{2}, 79--87 (1977/78).

\bibitem{Fri2} A. Frigerio: \textit{Stationary states of quantum dynamical semigroups.}
{\em Comm. Math. Phys.} \textbf{63}, 269--276 (1978).

\bibitem{FrVe} A. Frigerio, and M. Verri: \textit{Long-time asymptotic properties of dynamical semigroups on $W^*$--algebras.}
{\em Math. Z.} \textbf{180},  275--286 (1982).

\bibitem{Gr} U. Groh: \textit{ Positive semigroups on $C^*$- and $W^*$-Algebras.} In {\it One-parameter Semigroups of Positive Operators}, edited by R. Nagel. Lecture Notes in Mathematics 1184, (Springer-Verlag, Berlin, 1986); pp. 369--425.

\bibitem{FaRe1} F. Fagnola, and R. Rebolledo:\textit{On the existence of stationary states for quantum dynamical semigroups.} {\em J. Math. Phys.} \textbf{42}, 1296--1308 (2001).

\bibitem{FaRe2} F. Fagnola, and R. Rebolledo: \textit{Subharmonic projections for quantum Markov semigroup.} {\em J. Math. Phys.} \textbf{43}, 1074--1082 (2002).

\bibitem{FaRe3}  F. Fagnola, and R. Rebolledo: \textit{Notes on the Qulitative Behaviour  of Quantum Markov semigroups.} In {\it Open Quantum Systems III. Recent Developments}, edited by S. Attal, A. Joye, and C.--A. Pillet. Lecture Notes in Mathematics 1882. (Springer-Verlag, Berlin, 2006); pp. 161-206.

\bibitem{Um} V. Umanit\'a: \textit{Classification and decomposition of Quantum Markov Semigroups.} {\em Probab. Theory Relat. Fields} \textbf{134}, 603--623 (2006).
\bibitem{Fa} F. Fagnola: \textit{ Quantum Markov semigroups: structure and
asymptotics.} {\em Rend. Circ. Mat. Palermo serie II Suppl. No.}
\textbf{73}, 35--51 (2004).

\bibitem{Mo1} A. Mohari: \textit{Markov shift in non-commutative probability.} {\em J. Funct. Anal.} \textbf{199}, 189--209 (2003).

\bibitem{Mo2} A. Mohari: \textit{A resolution of quantum dynamical semigroups.} Preprint arXiv:math/0505384v1, May 2005.

\bibitem{BaNa} B. Baumgartner, and H. Narnhofer: \textit{Analysis of quantum semigroups with GKS--Lindblad
generators: II. General.} {\em J. Phys. A: Math. Theor.} \textbf{41}, 395303 (2008).

\bibitem{Sa} S. Sakai: \textit{$C^*$-algebras and $W^*$-algebras.} (Springer-Verlag, Berlin, 1971).

\bibitem{AsEl} L. Asimow, and A.J. Ellis: \textit{Convexity Theory and its Applications in Functional Analysis.} (Academic Press, London, 1980).

\bibitem{Za} P.R. Zangara: \textit{Evoluci\'on asint\'otica de sistemas cu\'anticos abiertos.} Trabajo Especial de Licenciatura en F\'{\i}sica, FaMAF-UNC, Diciembre 2009.
\end{thebibliography}
\end{document}